\documentclass{article}
\usepackage{ijcai23}

\usepackage{times}
\usepackage{soul}
\usepackage{url}
\usepackage[hidelinks]{hyperref}
\usepackage[utf8]{inputenc}
\usepackage[small]{caption}
\usepackage{graphicx}
\usepackage{amsmath}
\usepackage{amsthm}
\usepackage{amsfonts}
\usepackage{booktabs}
\usepackage{algorithm}
\usepackage[noend]{algorithmic}
\usepackage[switch]{lineno}
\usepackage{color}
\usepackage{multirow}
\usepackage{natbib}

\title{Average Envy-freeness for Indivisible Items}

\author{
Qishen Han$^1$
\and
Biaoshuai Tao$^2$\and
Lirong Xia$^{1}$
\affiliations
$^1$Rensselaer Polytechnic Institute\\
$^2$Shanghai Jiao Tong University
\emails
hnickc2017@gmail.com,
bstao@sjtu.edu.cn,
xialirong@gmail.com
}

\begin{document}

\newtheorem{thm}{Theorem}
\newenvironment{thmbis}[1]
  {\renewcommand{\thethm}{\ref{#1}}%
   \addtocounter{thm}{-1}%
   \begin{thm}}
  {\end{thm}}
\newtheorem{dfn}{Definition}
\newenvironment{dfnbis}[1]
  {\renewcommand{\thedfn}{\ref{#1}}%
   \addtocounter{dfn}{-1}%
   \begin{dfn}}
  {\end{dfn}}
\newtheorem{conj}{Conjecture}
\newtheorem{lem}{Lemma}
\newenvironment{lembis}[1]
  {\renewcommand{\thelem}{\ref{#1}}%
   \addtocounter{lem}{-1}%
   \begin{lem}}
  {\end{lem}}
\newtheorem{ex}{Example}
\newtheorem{Alg}{Algorithm}
\newtheorem{prob}{Problem}
\newtheorem{question}{Question}
\newtheorem{prop}{Proposition}
\newtheorem{coro}{Corollary}
\newenvironment{prof}{\noindent{\em Proof.}\rm }{\hfill $\Box$ }
\newenvironment{sketch}{\noindent{\em Proof sketch.}\rm }{\hfill $\Box$ }
\newtheorem{cond}{Condition}
\newtheorem{claim}{Claim}
\newtheorem{mes}{Message}
\newtheorem{view}{Viewpoint}
\newtheorem{calculation}{Calculation}
\newtheorem{obs}{Observation}
\newtheorem{RQ}{Research Question}
\newcounter{newct}
\newcommand{\rqu}{%
        \stepcounter{newct}%
        Research Question~\thenewct.}

\newcommand\qishen[1]{{\color{blue} \footnote{\color{blue}Qishen: #1}} }

\newcommand{\blue}[1]{\textcolor{blue}{#1}}

\newcommand{\gd}{g}
\newcommand{\vt}{v}
\newcommand{\ut}{u}
\newcommand{\aef}{\text{\rm AEF}}
\newcommand{\aefone}{\text{\rm AEF-1}}
\newcommand{\inst}{\mathcal{I}}
\newcommand{\mat}{H}
\newcommand{\nvec}{W}
\newcommand{\nvecset}{\mathcal{W}}
\newcommand{\matset}{\mathcal{H}}
\newcommand{\isvld}{Vld}
\newcommand{\prev}{Prev}
\newcommand{\rd}{r}
\newcommand{\ub}{a}
\newcommand{\rim}{removing matrix}
\newcommand{\Rim}{Removing matrix}
\newcommand{\rmat}{R}
\newcommand{\ritem}{M'}

\newcommand{\gds}{\gd^s}
\newcommand{\gdl}{\gd^l}

\newcommand{\aefoneplus}[1]{#1\text{-error }\aefone} 
\newcommand{\aefonemul}[1]{#1\text{-}\aefone} 

\newcommand{\exst}{\text{\sc Existence}}
\newcommand{\aefexst}{\text{{\rm AEF}-\exst}}
\newcommand{\aefoneexst}{\text{{\rm AEF-1}-\exst}}
\newcommand{\aefqexst}{\text{{\rm AEF}-\exst{} with a quota}}
\newcommand{\aefoneqexst}{\text{{\rm AEF-1}-\exst{} with a quota}}
\newcommand{\partition}{\text{\sc Partition}}
\newcommand{\ecpartition}{\text{\sc Equal-cardinality Partition}}

\maketitle

\begin{abstract}
    In fair division applications, agents may have unequal entitlements reflecting their different contributions. Moreover, the contributions of agents may depend on the allocation itself. Previous fairness notions designed for agents with equal or pre-determined entitlement fail to characterize fairness in these collaborative allocation scenarios. 

    We propose a novel fairness notion of {\em average envy-freeness} (AEF), where the envy of agents is defined on the average value of items in the bundles. Average envy-freeness provides a reasonable comparison between agents based on the items they receive and reflects their entitlements. We study the complexity of finding AEF and its relaxation, {\em average envy-freeness up to one item} (AEF-1). While deciding if an AEF allocation exists is NP-complete, an AEF-1 allocation is guaranteed to exist and can be computed in polynomial time. We also study allocation with {\em quotas}, i.e. restrictions on the sizes of bundles. We prove that finding AEF-1 allocation satisfying a quota is NP-hard. Nevertheless, in the instances with a fixed number of agents, we propose polynomial-time algorithms to find AEF-1 allocation with a quota for binary valuation and approximated AEF-1 allocation with a quota for general valuation.
\end{abstract}

\section{Introduction}
{Fair division aims to allocate items to a group of agents that have different preferences for the items and achieve fairness among the agents.
It is a classical yet heating topic that has wide application in real-world scenarios including peer review~\citep{payan2022order}, cloud computing~\citep{wang2015multi}, and healthcare resource distribution~\citep{roadevin2021can}. In fair division application, it is often the case that agents are {\em asymmetric}, i.e. having unequal entitlements. For example, in a food bank allocation, a family with more population requires more food to feed themselves. Asymmetric agents characterize a wide range of scenarios where agents have different contributions in a collaboration or agents represent groups with different populations. 

Extending fairness notions to agents with different entitlements, \citet{Chakraborty2021Weighted} proposes {\em weighted envy-freeness} (WEF), where the entitlements of agents are characterized by {\em predetermined weights}, and envy is defined on the weighted margin. Its relaxation, weighted envy-freeness up to one item (WEF-1), is guaranteed to exist and has been applied to paper-reviewer matching mechanisms~\citep{payan2022order}. 
However, in many scenarios, especially the collaboration between nations, groups, and individuals on productive activities, the entitlement and the contribution of agents depend on the allocation of resources. 

\begin{ex}
    Multiple research labs decide to start an interdisciplinary research collaboration. They need to allocate the research resources (e.g. funding, computing resources, assistants)  and assign tasks to each lab.
    The contribution (task) of a lab depends on the resource it gets, so both plans should be determined simultaneously. Different labs have expertise in different areas and have different preferences for resources. A computer science group would prefer more computing resources while a biological group would prefer more research assistants.
    Furthermore, labs want the resources allocated fairly based on their expected contributions. How should they fairly allocate the resources?
\end{ex}

There are two challenges in characterizing fairness in collaboration scenarios. Firstly, with heterogeneous backgrounds, it is likely agents have unequal entitlements due to their different contributions. Secondly, the allocation and the entitlements are decided simultaneously, and the entitlements of agents depend on the allocation. For example, we would expect a lab with more computing resources to make more contributions to the computational results. Under these challenges, neither EF nor WEF is able to characterize fairness in the collaboration scenarios. EF designed for equal entitlement cannot directly extend to different entitlements.  WEF requires a {\em predetermined} weight for each agent. When the entitlements change with the allocation, a fixed weight is unable to represent them.

Based on the challenges, the following question remains unanswered: {\bf what is a proper fairness notion for fair division under collaborations?}

\subsection{Average Envy-freeness}
We propose the notion of {\em average envy-freeness} ($\aef$), where envy between agents is defined on the average value of the items in the bundles. 
\begin{dfnbis}{dfn:aef}[Average envy-freeness ($\aef$)] An allocation $A$ is said to be average envy-free if for any pair of agents $i, h\in N$, $\frac{\vt_i(A_i)}{|A_i|}\ge \frac{\vt_i(A_h)}{|A_h|}$.
\end{dfnbis}
In average envy-freeness, the entitlement of each agent is the number of items they received. For example, in the paper review scenario, the contribution of a reviewer is evaluated by the number of papers they review. The reviewer with more papers has a larger contribution to the community and deserves a larger entitlement when allocating the papers.


The major difference between WEF and AEF is the modeling of the entitlements. In WEF, the entitlement of each agent is is their weight, which is pre-determined and independent of the allocation. In AEF, on the other hand, the entitlement or contribution of an agent is reflected by the size of their bundle. In a collaboration scenario where the allocation and the entitlements are decided simultaneously, predetermined weights cannot characterize entitlements that depend on the allocation. On the other hand, when there is no large difference between entitlements brought by different items, the size of each bundle serves as a natural estimation of the entitlements of each agent.  

We believe that average envy-freeness is a proper fairness notion for the collaboration scenarios because it provides a reasonable comparison between agents based on the resources they received. In an AEF allocation, every agent believes that they get a fair share of resources based on their (expected) contribution to the resource.

\subsection{Our contribution}

We study the existence and the complexity of average envy-freeness and average envy-freeness up to one item ($\aefone$). We also consider scenarios with {\em quotas}, i.e. restrictions on the size of the bundles. The quota reflects the requirements and capability of each agent in the allocation. For example, a research lab needs a minimum amount of funding to run experiments and can digest at most a maximum amount of funding. 

The summary of our result is shown in Table~\ref{tbl:result}. We first consider the existence of $\aef$ and $\aefone$ without quota. Deciding whether an $\aef$ allocation exists is NP-complete, while an $\aefone$ allocation always exists and can be computed in polynomial time. However, the problem becomes much more difficult when quotas are considered, as deciding whether there exists an $\aefone$ allocation satisfying the quota is NP-complete. Therefore, we consider the instances with a fixed number of agents. When the value of each item is either zero or one, we propose a polynomial-time dynamic programming algorithm to decide the existence of an $\aefone$ allocation satisfying the quota and find the allocation if it exists. On the other hand, deciding the existence of an $\aefone$ allocation satisfying the quota in general valuation is still NP-complete. For the general valuation case, we give an approximation algorithm that finds an $\aefonemul{(1 - \frac{4}{mn})}$ allocation, in which agents value their bundles at least $(1 - \frac{4}{mn})$ times of other agent's bundles on the average value after removing one item. 

\begin{table}[htbp]
\setlength\tabcolsep{5.5pt}
\begin{tabular}{@{}cccc@{}}
\toprule
\textbf{Criterion} & $n$                                & \textbf{Valuation} & \textbf{Result }       \\ \midrule
$\aef$ &
  \begin{tabular}[c]{@{}c@{}}$2$\\ General\end{tabular} &
  \begin{tabular}[c]{@{}c@{}}Identical\\ General\end{tabular} &
  NP-complete \\ \midrule
$\aefone$ & General                            & General   & Always Exists \\ \midrule
\multirow{5}{*}{\begin{tabular}[c]{@{}c@{}}$\aefone$ \\ with quota\end{tabular}} &
  General &
  \begin{tabular}[c]{@{}c@{}}Binary\\ General\end{tabular} &
  NP-complete \\ \cmidrule(l){2-4} 
          & \multirow{3}{*}{Constant, $\ge 3$} & Binary    & In P          \\ \cmidrule(l){3-4} 
 &
   &
  General &
  \begin{tabular}[c]{@{}c@{}}NP-complete\\ $(1-\frac{4}{mn})$-approx\end{tabular} \\ \bottomrule
\end{tabular}
\caption{Result Summary of $\aef$ and $\aefone$. $n$ and $m$ are the number of agents and the number of items respectively.\label{tbl:result}}
\end{table}

\section{Related Works}

There is a large literature on the fair division problem with asymmetric agents, i.e. individuals or groups with heterogeneous entitlements.
\citet{Chakraborty2021Weighted} proposes the notion of {\em weighted envy-freeness up to one item} (WEF-1), where predetermined weights of each agent represent their entitlements, and shows that a WEF-1 allocation always exists. \citet{payan2022order} and ~\citet{Chakraborty2021Picking} focus on maximizing social welfare of WEF-1 allocations by select the correct picking sequence in a Round-Robin-like mechanism.  \citet{Chakraborty2022Revisited} proposes a parameterized family of weighted envy-freeness, where agents with large or small weights are favored by setting different parameters. 
Other fairness notions on asymmetric agents include weighted MMS~\citep{farhadi2019fair,Babaioff2021Fair,Aziz2020Proportional}, weighted proportionality~\citep{Aziz2020Proportional,Li2022Almost}, and maximizing weighted Nash social welfare~\citep{Suksompong2022Nash}.

Another related line of work is fair division with {\em cardinality constraints}. \citet{Biswas18Cardinality} considers a scenario where items are categorized into multiple groups, and each group has a capacity of contributing to each agent. They show that an EF-1 allocation and a constant-factor approximation of MMS allocation always exist in such a scenario. They also extend this type of constraint to be represented by a matroid. 
\citet{Biswas2019Matroid} designs a polynomial time algorithm to find an EF-1 allocation satisfying the matroid constraint. \citet{Dror2021Heterogeneous} studies scenarios where agents have heterogeneous matroid constraints, and provides an algorithm to find EF-1 allocations in certain circumstances. \citet{gan2021approximately} and~\citet{Babaioff2019Income} study fair division problem where agents have budget constraints on the items. \citet{aziz2019constrained} propose a mechanism that turns a welfare-efficient allocation into a fair allocation while preserving the efficiency constraint.

}

\section{Preliminaries}
{\paragraph{Problem Instance} An instance of fair division problem $\inst= \langle N, M, V \rangle$ is defined by a set of $n$ agents $N = [n]$, a set of $m$ items $M$, and a valuation profile $V = \{\vt_1, \vt_2,\cdots, \vt_n\}$. We use $i$ to denote a generic agent in $N$ and $\gd$ to denote a generic item in $M$.

\paragraph{Valuation and Average Value}
Valuation functions represent the preferences of agents among the items. For each agent $i$, $\vt_i$ a mapping to a subset of $M$ to a non-negative value $\vt_i: 2^{M}\to \mathbb{R}_{\ge 0}$. We follow the convention to assume {\em additive valuation}, i.e for any $i\in N$ and $M'\subseteq M$, $\vt_i(M') = \sum_{\gd\in M'} \vt_i(\gd)$. 
We say that an instance has {\em binary} valuation if for any $i\in N$ and $\gd\in M$, $\vt_i(\gd) \in\{0, 1\}$, and an instance has {\em identical} valuation if $\vt_1 = \vt_2 = \cdots = \vt_n$. 
We also define the {\em average value} of a subset of item $M'$ as $\ut_i(M') = \frac{\vt_i(M')}{|M'|}$, i.e. the average value of the items in the subset. Specifically, $\ut_i(\emptyset) = 0$.  


\paragraph{Allocation} An allocation $A = (A_1, A_2,\cdots, A_n)$ is a $n$-partition of the set of items $M$, where $A_i\subseteq M$ is the {\em bundle} allocated to agent $i$.  We sometimes abuse the notation and use $A$ to denote a partial allocation, where there are items unallocated to any agent.

\paragraph{Quota} A quota $Q$ is a constraint on allocations. For each agent $i$, $Q$ imposes an upper bound and a lower bound for the size of the bundle $A_i$. An allocation $A$ satisfies a quota $Q$ if all agents satisfy the constraint. A quota is said to be {\em exact} if the upper bound equals to the lower bound for every agent. An exact quota regulates the exact number of items in each bundle. 

\begin{dfn}[Average envy-freeness ($\aef$)] 
\label{dfn:aef}An allocation $A$ is said to be average envy-free if for any pair of agents $i, h\in N$, $\ut_i(A_i)\ge \ut_i(A_h)$.
    
\end{dfn}



\begin{dfn}[Average envy-freeness up to one item ($\aefone$)]
    An allocation $A$ is said to be average envy-free up to one item if for any pair of agents $i, h\in N$, there exists an item $\gd\in A_i\cup A_h$ such that $\ut_i(A_i\setminus \{\gd\})\ge \ut_i(A_h\setminus \{\gd\})$.
\end{dfn}

By the definition of $\aefone$, agents can remove an item from either bundle under comparison. This is because agents can increase the average value of their own bundle by removing the least preferred item (if it's not the only item). It follows from the definition that an $\aef$ allocation is always $\aefone$.  

We introduce the computational problems related to $\aef{}$ and $\aefone{}$. 

\begin{dfn}[$\aefexst$]
    Given an instance $\inst$, does there exist an allocation $A$ such that $A$ is an $\aef$ allocation?
\end{dfn}

\begin{dfn}[\aefoneexst]
    Given an instance $\inst$, does there exist an allocation $A$ such that $A$ is an $\aefone$ allocation?
\end{dfn}

\begin{dfn}[$\aefqexst$]
    Given an instance $\inst$ and a quota $Q$, does there exist an allocation $A$ such that $A$ is an $\aef$ allocation and satisfies $Q$?
\end{dfn}

\begin{dfn}[$\aefoneqexst$]
    Given an instance $\inst$ and a quota $Q$, does there exist an allocation $A$ such that $A$ is an $\aefone$ allocation and satisfies $Q$?
\end{dfn}}

{\section{Average Envy-freeness without quotas}

This section focuses on finding $\aef$ and $\aefone$ allocations without a quota constraint. We show that deciding the existence of an $\aef$ allocation is NP-complete, while an $\aefone$ allocation always exists and can be found in polynomial time. 

\begin{thm}
\label{thm:aef}
    $\aefexst$ is NP-complete even for two agents with identical valuations. 
\end{thm}

\begin{proof}[Proof Sketch]
We construct a reduction from $\partition$, which is known to be NP-complete~\citep{Garey79:Computers}. An instance of $\partition$ consists of a multiset $X = \{x_1, x_2,\cdots, x_{k}\}$ where $x_i\in \mathbb{N}$. The goal is to determine whether $X$ can be partitioned into two subsets $Y$ and $X\setminus Y$ with equal sum $T$.

Given an instance of $\partition$ $X$, we construct an $\aefexst$ instance with two agents and $m =2k$ items. Two agents share the same value function $\vt$, and $\ut$ is the average valuation function. For each $x_i$, there exists two items $\gds_i $ and $ \gdl_i$ such that $\vt(\gds_i) = (T^2k^2)^i$, and $\vt(\gdl_i) = (T^2k^2)^i + x_i$. 

$(\Rightarrow)$ Given $Y$ and $X\setminus Y$ being an equal-sum partition of $X$, we construct an allocation $A$. For each $i$, $A_1$ gets $\gdl_i$ if $x_i \in Y$ or $\gds_i$ if $x_i\in X\setminus Y$. $A_2$ gets the rest of the items. It is not hard to verify that both agents get $k$ items, and $\vt(A_1) = \vt(A_2)$. Therefore, $\ut(A_1) = \ut(A_2)$, and $A$ is an $\aef$ allocation. 

$(\Leftarrow)$ If there exists an $\aef$ allocation $A$, we show that $A$ induces an equal-sum partition of $X$ in four steps. 

First, each agent gets exactly one of the largest items $\gds_k$ and $\gdl_k$. Otherwise, the exponential term $(T^2k^2)^i$ guarantees that the agent without the two items envies the other agent. 

Second, each agent gets exactly $k$ items. $\gds_k$ and $\gdl_k$ guarantee that the additive value of bundles is at the same level, and the agent with more items envies the agent with fewer items. 

Third, for each $i = 1,2,\cdots, k$, each agent get exactly one of $\gds_i$ and $\gdl_i$, following the similar reasoning of the first step. 

Finally, the allocation induces an equal-sum partition of $X$. Let $Y = \{x_i | \gdl_i\in A_1\}$, and $Y$ and $X\setminus Y$ is an equal-sum partition of $X$.  Given that each agent get exactly one of $\gds_i$ and $\gdl_i$, the difference between two bundles just come from $x_i$ from each $i$. Therefore, the sum of $x_i$ in $A_1$ must equal to the sum of $x_i$ in $A_2$, which implies $Y$ and $X\setminus Y$ be an equal-sum partition of $X$. The full proof is in Appendix~\ref{apx:aef}.
\end{proof}

The fact that $\aefexst$ is already NP-complete without quota directly implies that $\aefexst$ with quotas is NP-complete. 

\begin{coro}
    $\aefqexst$ is NP-complete even for two agents with identical valuations. 
    
\end{coro}

Despite that $\aef$ allocation is hard to find, we show that $\aefone$ allocation always exists just like EF-1. 

\begin{prop}
    For any instance $\inst$, an $\aefone$  allocation always exists and can be found in polynomial time. 
\end{prop}

\begin{proof}
    Consider the following allocation scheme:
    \begin{itemize}
        \item If $m\le n$, agents $1,2,\cdots, m$ get their favorite item among the unallocated items in turns. The rest agents get nothing. 
        \item If $m > n$, $1,2,\cdots, n-1$ get their favorite item among the unallocated items in turns, and agent $n$ gets the rest of the items. 
    \end{itemize}
    We show that allocation induced by this scheme is $\aefone$. Consider two agents $i$ and $h$. We show that $i$ does not envy $h$ up to one item.  
    \paragraph{$m\le n$.} If $i\le m$, then $i$ does not envy any $h > i$ because $h$ gets either no item or an item inferior to $i$'s item under $i$'s valuation. If $h < i$, then $i$ and does not envy $h$ by removing $h$'s only item. If $i > m$, then $i$ does not envy any other $h$ after removing $h$'s item (if exists). 

    \paragraph{$m > n$.} If $h\neq n$, $i$ does not envy $h$ after removing $h$'s only item. If $h = n$, note that $i$ picks their favorite item among all the rest of the items, including all $h$'s items. Therefore, for any $\gd\in A_{h}$, $\vt_{i}(A_i) \ge \vt_{i}(\gd)$. Therefore, $\ut_{i}(A_{i}) \ge \ut_{i}(A_{h})$, and $i$ does not envy $h$. 
\end{proof}

\section{AEF-1 with a quota}
In this section, we focus on the complexity of $\aefone$-{\sc Existence} with quota. Although an AEF-1 allocation always exists, it is not likely in real-world applications that all but one agent get exactly one item, and the rest agent gets all the rest items. Quotas restrict the size of each bundle in reasonable ranges and lead to allocations reflecting real-world scenarios. Unfortunately, our first result shows that $\aefone$-{\sc Existence} with a quota is NP-complete even for binary valuations. 

\begin{thm}
    $\aefoneqexst$ is NP-com-plete even for binary valuations. 
\end{thm}

\begin{proof}
    We show a reduction from {\sc EF-Existence} for binary valuations, which is known to be NP-complete~\citep{aziz2015fair,Hosseini2020Withholding}.

    An {\sc EF-Existence} with binary value instance consists of a set of agents, $N =[n]$, a set of items $M$ ($|M| = m$), and a binary additive valuation profile $V$. The goal is to determine whether there is an envy-free allocation. 
    
    We construct a $\aefoneqexst$ instance as follows: $N' = N$, $M' = M \cup D$, where $D$ is a set of $(n-1)m$ items which have no value to any agent.  Additive valuation profile $V'$ is defined as follows: for each agent $i$, and item $g$, if $g \in M$, $\vt_i'(g) = \vt_i(g)$; otherwise, $\vt_i'(g)=0$. $\vt_i'(\emptyset) = 0$. $\ut'$ is the average value function of $\vt'$. The quota $Q$ requires every agent to receive exactly $m$ items.

    $(\Rightarrow)$ Suppose {\sc EF-Existence} is a YES instance, and $A^*$ is an envy-free allocation under $V$. We show that $\aefoneqexst$ is a YES instance. Let $A'$ be an allocation in the $\aefoneqexst$ instance where each agent $i$ gets all the items in  $A^*_i$ and fills up the quota with items in $D$. It follows from the definition that $A'$ satisfies $Q$. Now we show that $A'$ is $\aefone$. Note that for any agents $i, h$, $\ut_{i}'(A_{h}') = \frac1m \vt_{i}(A^*_{h})$. Since $A^*$ is envy-free, $\vt_{i}(A^*_{i})\ge \vt_{i}(A^*_{h})$. Therefore, $\ut_{i}'(A_{i}')\ge \ut_{i}'(A_{h}')$, and $A'$ is an $\aef$ (thus $\aefone$) allocation.

    $(\Leftarrow)$ Suppose $\aefoneqexst$ is a YES instance, and $A'$ is an $\aefone$ allocation satisfying $Q$. We first show that $A$ must also be an $\aef$ allocation. Suppose this is not the case, and agent $i$ envies agent $h$. We show $i$ envies $h$ even after removing one item. From binary valuation, we have $\vt_i'(A_i') \le \vt'_i(A_h') -1$. If agent $i$ removes one item from $A_h$, the average value of $A_h'$ is not smaller than $A_i'$, and $A_h'$ has fewer items than $A_i'$. Therefore, $A_h'$ still has a higher average value than $A_i'$, and $i$ envies $h$. If agent $i$ removes one item from $A_i'$, the average value will be no more than $\frac{\vt_i'(A_h')-1}{m-1}$. This value is strictly less than $\ut_i(A_h') =\frac{\vt_i'(A_h')}{m}$ for $\vt_i'(A_h') < m$. If $\vt_i'(A_h') = m$, then $\vt_i'(A_i') =0$ since $M'$ contains at most $m$ valuable items for $i$. Therefore, $i$ still envies $h$ after removing any item. This is a contradiction. Therefore, $A'$ must also be $\aef$ allocation.

    Now we show that {\sc EF-Existence} with binary value is also a YES instance. Let $A^*$ be a allocation in {\sc EF-Existence} instance such that $A^*_i = A'_i \cap M$ for every $i\in N$.  Similarly, with relationship $\ut_{i}'(A_{h}') = \frac1m \vt_{i}(A^*_{h})$, the $\aef$-ness of $A'$ implies the envy-freeness of $A^*$. 
\end{proof}

Due to the hardness of the problem, we turn to consider $\aefone$ allocation with a {\em fixed} number of agents $n$. We show that, for binary valuations, $\aefoneqexst$ for a fixed number of agents is in $P$, in contrast with the hardness in the general $n$ case.  

\begin{thm}
    There exists a polynomial-time algorithm that, given any instance of $\aefoneqexst$ for a fixed number of agents and binary valuations, decides if there exists an $\aefone$ allocation satisfying the quota, and outputs an allocation if there exists. 
\end{thm}

~\citep{Aziz2022Computing} proposes a pseudo-polynomial time dynamic programming algorithm to find an EF-1 allocation maximizing social welfare given a constant number of agents. We apply their technique and propose
Algorithm~\ref{alg:dp1} to compute $\aefoneqexst$ for binary valuations.

\paragraph{State} A state in Algorithm~\ref{alg:dp1} is a triplet $(\nvec, \mat, k)$. Suppose $A$ is a partial allocation of $M = \{\gd_1, \gd_2, \cdots, \gd_m\}$ where $g_1, g_2,\cdots, g_k$ has been allocated. $\nvec$ is a $n$-vector that records the number of items each agent is allocated, i.e. $\nvec_i = |A_i|$ for each $i$. $\mat$ is a $n\times n$-matrix that records the additive value of each agent toward each bundle, i.e. $\mat(i, h) = \vt_i(A_h)$. $\nvec$ and $\mat$ together record each agent's average value on each bundle. For any pair of agents $i, h$, $\ut_i(A_h) = \frac{\mat(i, h)}{\nvec_h}$. $k = 0,1,2\cdots, m$ indicates that item $\gd_1,\gd_2, \cdots, \gd_k$ has been allocated while other items are not. $k=0$ means no item has been allocated yet. For each state, we maintain two values. $\isvld(\nvec, \mat, k)\in \{0,1 \}$ indicates whether this state is reached, which stands for there exists a partial allocation of whom the state is $(\nvec, \mat, k)$. $\prev(\nvec, \mat, k)\in N$ records the agent that item $\gd_k$ is allocated to reach the current state. For any allocation, it is sufficient to judge whether it is $\aefone$ and satisfies $Q$ from its corresponding state.

\paragraph{State Transition} For a given state $(\nvec, \mat, k)$ with $k < m$ and $\isvld(\nvec, \mat, k) = 1$, we enumerate the agent to whom item $\gd_{k+1}$ is allocated. For each agent $i$, we find the updated state $(\nvec', \mat', k+1)$ after $\gd_{k+1}$ is allocated to $i$ and set $\isvld(\nvec', \mat', k+1) = 1$ and $\prev(\nvec', \mat', k+1) = i$. The algorithm start from $(\mathbf{0}^n, \mathbf{0^{n\times n}}, 0)$ and iterated for $k = 0,1,\cdots, m-1$. The search space of $\nvec$ is $\nvecset = \{0,1,\cdots, m\}^n$, and the search space $\mat$ is $\matset = \{0, 1,\cdots, m\}^{n\times n}$. Finally, the algorithm finds if there is a state $\isvld(\nvec, \mat, m) = 1$ that is $\aefone$ and satisfies $Q$. If so, the algorithm outputs YES and constructs the allocation backward with $\prev$. Otherwise, the algorithm outputs NO. 

\begin{algorithm}[htbp]
\caption{DP for AEF-1 with quota with binary valuation}
\label{alg:dp1}
\begin{algorithmic}[1]
\REQUIRE Agent set $N$, Item set $M$, binary valuation profile $V$, and quota $Q$. 
\ENSURE An $\aefone$ allocation satisfying $Q$ if it exists. 
\STATE Initialization: $\isvld(\mathbf{0}^n, \mathbf{0}^{n\times n}, 0) \leftarrow 1$.
\FOR{$k=0,1,\cdots, m$}
\FOR{$\nvec \in \nvecset$ and $\mat\in\matset$ such that $\isvld(\nvec, \mat, k) = 1$}
\FOR{$i = 1,2,\cdots n$}
\STATE Update $\nvec', \mat' $ after assigning $\gd_{k+1}$ to $i$.
\STATE $\isvld(\nvec', \mat', k+1) \leftarrow 1$.
\STATE $\prev(\nvec,\mat', k+1) \leftarrow i$.
\ENDFOR
\ENDFOR
\ENDFOR
\FOR{$\nvec\in \nvecset, \mat\in\matset$ such that $\isvld(\nvec,\mat, m) = 1$}
\IF{$(\nvec,\mat, m)$ is \aefone{} and satisfies $Q$}
\STATE Construct the allocation from $\prev$ backward 
\STATE {\bf return} the allocation.
\ENDIF
\ENDFOR
\RETURN NO
\end{algorithmic}
\end{algorithm}


The technique of enumerating all possible values in Algorithm~\ref{alg:dp1} can be extended to valuations where a bundle has at most $Poly(m)$ different values. However, for general valuation, a bundle can have exponentially many values. In fact, we show that $\aefoneqexst$ with fixed $n \ge 3$ is NP-complete. 

\begin{thm}
    \label{thm:aefoneq}
    $\aefoneqexst$ with fixed $n\ge 3$ is NP-complete. 
\end{thm}

\begin{proof}[Proof Sketch]
    We propose a reduction from the computation problem of  {\sc Equal-cardinality Partition}, a variation of $\partition$ that requires equal size between two subsets and is also NP-complete~\citep{Garey79:Computers}. An instance of {\sc Equal-cardinality Partition} consists of a multiset $X = \{x_1, x_2,\cdots, x_{2k}\}$ where $x_i\in \mathbb{N}$. The goal is to determine whether $X$ can be partitioned into two subsets $Y$ and $X\setminus Y$ with equal size $k$ and equal sum $T$.

    Given an {\sc Equal-cardinality Partition} instance, we construct a $\aefoneqexst$ instance with three agents and $3k+6$ items. (If $n > 3$, we add agents that value all items as 0 and are required to receive no items by the quota.) Agents share the same valuation function $\vt$ and average value function $u$. The quota $Q$ requires every agent to be allocated exactly $k+2$ items. 
     The item set $M = M_1\cup M_2\cup M_3$ consists of three parts:
    \begin{itemize}
        \item $M_1 = \{\gd_1, \gd_2, \cdots, \gd_{2k}\}$, where $\vt(\gd_j) = x_j + k^2T^2$. \\Let $T' =\frac12 \sum_{\gd\in M_1} \vt(\gd) =T + k^3T^2.$
        \item $M_2$ contains $k+1$ copies of $b$ with $\vt(b) = \frac{(k+2)T'}{(k+1)^2}.$
        \item $M_3$ contains five copies of $0$ with $\vt(0) = 0$. 
    \end{itemize}

    We state that the value of $b$ is smaller than any item in $M_1$. 
    \begin{lem}
        \label{lem:aefq}
        For any $g\in M_1$, $\vt(b) < \vt(\gd).$
    \end{lem}
    
    $(\Rightarrow)$ If {\sc Equal-cardinality Partition} is a YES instance, and $Y$ and $(X\setminus Y)$ are a equal-size and equal-sum partition, we show the following allocation $A$ is $\aefone$ and satisfies $Q$.
    \begin{enumerate}
        \item $A_1 = \{g_j\mid x_j\in Y\} \cup \{0,0\}$.
        \item $A_2 = \{g_j\mid x_j\in (X\setminus Y)\}\cup \{0,0\}$.
        \item $A_3 = M_2\cup \{0\}$. 
    \end{enumerate}
    It's not hard to verify that each agent gets exactly $k+2$ items, $\ut(A_1) = \ut(A_2) = \frac{T'}{k+2}$, and $\ut(A_3) = \frac{T'}{k+1}$. Agent 1 and 2 does not envy each other, and agent 3 does not envy agent 1 and 2. When comparing with agent 3, agent 1 and agent 2 can remove an item 0 in their own bundle, and $\ut(A_1\setminus \{0\}) = \ut(A_2\setminus \{0\}) = \ut(A_3) = \frac{T'}{k+1}$. Therefore, $A$ is $\aefone$ and satisfies $Q$. 

    $(\Leftarrow)$ If $\aefoneqexst$ is a YES instance, and $A$ is an $\aefone $ allocation satisfying $Q$. We show that {\sc Equal-cardinality Partition} is a YES instance in three steps. 
    
    First, no agent can have more than two item $0$ in their bundles. Otherwise, the agent get at least three 0 envies the agent get at most one 0 even after removing one item. 
    
    Second, the agent with exactly one item $0$ (agent 3, with loss of generality) must have all the item $b$. Otherwise, since $\vt(b) < \vt(\gd)$ for any $\gd\in M_1$, the average value of $A_3$ will exceed $\frac{T'}{k+1}$, and one of $A_1$ and $A_2$ will have average value strictly less than $\frac{T'}{k+2}$. Then the owner of this bundle will envy agent 3 even after removing one item. 
    
    Finally, agent 1 and agent 2's bundles must derive a equal-cardinality partition of $X$. Otherwise, the average value of the less-valuable bundle will be strictly less than $\frac{T'}{k+2}$. With the same reasoning as the second step, the owner of this bundle will envy agent 3 (with $\ut(A_3) = \frac{T'}{k+1}$) even after removing one item. 
    Therefore, {\sc Equal-cardinality partition} is a YES instance. The full proof is in Appendix~\ref{apx:aefoneq}
\end{proof}

}

\section{Approximation on AEF-1 with a quota}
{The NP-hardness on $\aefoneqexst$ urges us to look into approximation results. A natural idea is to round the value of items so that a bundle can have at most $Poly(m)$ different values and apply the procedure of Algorithm~\ref{alg:dp1}.

For simplicity of calculation, we assume $\max_{i,\gd} \vt_i(\gd) = 1$. We propose two approximation notions of $\aefone{}$ based on the additive error and multiplicative ratio respectively. 

\begin{dfn}[$\aefoneplus{\varepsilon}$]
    Given $\varepsilon \ge 0$, an allocation $A$ is $\aefoneplus{\varepsilon}$ if for any pairs of agents $i, h\in N$, there exists an item $\gd\in A_i\cup A_h$ such that $\ut_i(A_i \setminus \{\gd\}) \ge  \ut_i(A_h \setminus \{\gd\}) - \varepsilon$. 
\end{dfn}

\begin{dfn}[$\aefonemul{\alpha}$]
    Given $0<\alpha \le 1$, an allocation $A$ is $\aefonemul{\alpha}$ if for any pairs of agents $i, h\in N$, there exists an item $\gd\in A_i\cup A_h$ such that $\ut_i(A_i \setminus \{\gd\}) \ge \alpha\cdot \ut_i(A_h \setminus \{\gd\})$.
\end{dfn}

\begin{prop}
   Given the assumption that $\max_{i,\gd} \vt_i(\gd) = 1$, if an allocation is $\aefonemul{\alpha}$, then it is $\aefoneplus{(1-\alpha)}$.
\end{prop}
\begin{proof}
    From $\aefonemul{\alpha}$ we know that $\ut_i(A_i \setminus \{\gd\}) \ge \alpha\cdot \ut_i(A_h \setminus \{\gd\})$. Therefore,
    \begin{eqnarray*}
        \ut_i(A_i \setminus \{\gd\}) &\ge &\alpha\cdot \ut_i(A_h \setminus \{\gd\})\\
        &= & \ut_i(A_h \setminus \{\gd\}) - (1-\alpha)\ut_i(A_h \setminus \{\gd\})\\
        &\ge & \ut_i(A_h \setminus \{\gd\}) - (1-\alpha).
    \end{eqnarray*}
    The last inequality comes from $ \ut_i(A_h \setminus \{\gd\}) \le 1$. 
\end{proof}

Proposition 2 tells us that $\aefonemul{\alpha}$ implies $\aefoneplus{\varepsilon}$ given bounded valuations. However, $\aefoneplus{\varepsilon}$ does not guarantee $\aefoneplus{\varepsilon}$, as shown in the followings (Example~\ref{ex:rounding}). Our goal is to find an approximation algorithm that returns an $\aefonemul{\alpha}$ with $\alpha$ close to 1 if possible. We first introduce our rounding scheme.


\paragraph{Rounding} 
Given the rounding parameter $\rd \in \mathbb{N}^+$ and an upper bound $\ub > 0$, we divide $[0, \ub]$ into $r+1$ intervals $\{0\}, (0, \frac{\ub}{r}], (\frac{\ub}{r}, \frac{2\ub}{r}],\cdots,$ $ (\frac{(r-1)\ub}{r}, \ub]$.  For $k=1,2,\cdots r$, a positive value $\frac{(k-1)\ub}{\rd} < x \le \frac{k\ub}{r}$ is rounded to $\frac{k\ub}{r}$. $0$ is rounded to $0$. 

If we apply the same rounding scheme to each $\vt_i(\gd)$ and directly apply Algorithm~\ref{alg:dp1}, we will be able to find an $\aefoneplus{\frac{2\ub}{\rd}}$ allocation, because an $\aefone$ allocation in the original valuations implies an $\aefoneplus{\frac{\ub}{\rd}}$ allocation in the rounded valuations, and an $\aefoneplus{\frac{\ub}{\rd}}$ allocation in the rounded valuations in turn implies an $\aefoneplus{\frac{2\ub}{\rd}}$ allocation in the original valuations. However, there is no guarantee on $\aefonemul{\alpha}$, because the value of an item can be rounded from arbitrarily small to $\frac{\ub}{\rd}$. 
Example~\ref{ex:rounding} shows a case where an $\aefone$ allocation in the rounded valuation turns out to be a poor approximation in the original valuation.

\begin{ex}
\label{ex:rounding}
Given any rounding parameters $\ub, \rd$ (assuming $\ub < \rd$), consider an instance with more than two agents and more than three items.
The table below describes an allocation where the first two agents $1$ and $2$ are allocated the first three items $g_1,g_2,g_3$.


\begin{table}[htbp]
\centering
\label{tbl:rounding}
\renewcommand\arraystretch{1.5}
\setlength\tabcolsep{10pt}
\begin{tabular}{@{}c|ccc@{}}
\toprule
  & $g_1$ & $g_2$         & $g_3$                     \\ \midrule
1 & {\huge \textcircled{\normalsize $\frac{a}{r}$}}    & {\huge\textcircled{\normalsize$\frac{a}{r}$}} & $\frac{\varepsilon a}{r}$ \\
2 & $\frac{a}{r}$     & $\frac{a}{r}$ & {\huge\textcircled{\normalsize$\frac{\varepsilon a}{r}$}} \\ \bottomrule
\end{tabular}%
\caption{Allocation $A$ where rounding leads to poor approximation.}
\end{table}

For item $\gd_3$, $\varepsilon\in (0,1)$ is an arbitrarily small positive value. It is not hard to verify that $A$ is not $\aefone$, as agent 2 envies agent 1 even after removing an item. Moreover, the allocation is no better than $\aefonemul{\varepsilon}$.  However, after rounding, the value of $\gd_1$ and $\gd_2$ is unchanged while the value of $\gd_3$ is rounded to $\frac{\ub}{\rd}$. Then $A$ is $\aefone$ in the rounded valuation. 
Therefore, if Algorithm~\ref{alg:dp1} finds (the state of) $A$, it returns an $\aefonemul{\varepsilon}$ allocation where $\varepsilon$ can be arbitrarily small. 

Although the instance has an $\aefone$ allocation $A'$ where agent 1 gets $\gd_1$ and $\gd_3$ and agent 2 gets $\gd_2$, the algorithm may not find $A'$. Note that after rounding $\gd_2$ and $\gd_3$ both have value of $\frac{\ub}{\rd}$. This means that $A$ and $A'$ share the same state in Algorithm~\ref{alg:dp1}. Which allocation is constructed depends on the $\prev$ record. By carefully manipulating the order of items, we can let the algorithm returns $A$ rather than $A'$. 
\end{ex} 

Therefore, we need a more refined rounding and searching scheme that can distinguish between $A$ and $A'$ to ensure a closer approximation ratio between the original and the rounded valuation. The rounding of each agent should be proportional to their valuations so that the rounding error is not too large compared with the value of their own bundles.
\citet{Bu2022Welfare} proposes a bi-criteria approximation algorithm to maximize EF-1 ratio and social welfare simultaneously. We follow their techniques to enumerate all items being removed in the envy comparisons, i.e. the ``1'' in ``$\aefone$''. 


\paragraph{\Rim} A \rim{} $\rmat$ is a matrix recording the items to remove when agents compare bundles with each other. For every pair of agents $i, h$, $\rmat(i, h) = (\gd, l) \in (M \cup \{\emptyset\}) \times \{i, h\}$. The first value $\gd$ is the item to remove when agent $i$ compares their bundle with $h$'s bundle, and the second value $l$ indicates whether $\gd$ belongs to $i$ or $h$.
$\gd = \emptyset$ means $i$ does not remove any item when comparing with $h$. In this case, $l$ makes no difference. Specifically, $\rmat(i, i) = (\emptyset, i)$ for each agent $i$. A \rim{} is valid if it derives a partial allocation of $M$. That is, it does not contain two entries $(\gd, l_1)$ and $(\gd, l_2)$ such that $l_1\neq l_2$. 

Our algorithm runs in four steps. We enumerate on all valid \rim{} $\rmat$. For each $\rmat$, we first allocate the items that have been pre-allocated by $\rmat$. Next, we round the values of the unallocated items based on each agent's valuation. Then, we run the dynamic programming search on the unallocated items under the rounded valuation to validate all possible states. Finally, we search states where an agent will not envy another agent by their rounding error under the rounded valuation after removing one item. If such a state exists, the algorithm returns the corresponding allocation. Otherwise, the algorithm returns NO. A detailed description of the algorithm is in Appendix~\ref{apx:dp2_alg}. 

\paragraph{Step 1: pre-allocation}
Given an \rim{} $\rmat$, let $M^{\rmat}$ be the set of item $\gd$ such there there exists an entry $(\gd, i)\in \rmat$. We allocate items in $M^{\rmat}$ according to $\rmat$. Let $\nvec_0^{\rmat}$ and $\mat_0^{\rmat}$ be the vector of bundle sizes and the valuation matrix after $M^{\rmat}$ has been allocated. 

\paragraph{Step 2: rounding}
For each agent $i$, let $M^{\rmat}_i$ be the set of $i$'s removing items, and 
Let $\ritem_i = M \setminus M^{\rmat}_i$. For each agent $i$, we set the rounding upper-bound $\ub = \max_{\gd\in \ritem_i} \vt_i(\gd)$. $\rd = m^2n^2$ is the same for all agents. We create the rounded valuation $\vt^{\rmat}_i$ by rounding the values for all items in $\ritem_i$. That is, for any agent $i$, $\vt_i^{\rmat}(\gd) = \vt_i(\gd)$ if $g\in M^{\rmat}_i$, and $\vt_i^{\rmat}(\gd)$ is rounded to the closest larger $\frac{k\ub_i}{\rd}$ if $g\in \ritem_i$. Let $\ut^{\rmat}_i$ be the average value function of $\vt^{\rmat}_i$. Given a fixed $M^{\rmat}_i$, each bundle will have at most $Poly(m)$ possible values. Precisely, $v^{\rmat}_i$ takes value from $\{0, \frac{\ub_i}{\rd}, \frac{2\ub_i}{\rd},\cdots, |\ritem_i|\cdot \ub_i\}$.

\paragraph{Step 3: dynamic programming}
The dynamic programming follows the same procedure of Algorithm~\ref{alg:dp1} that enumerates and iteratively validates states. The difference is that we run the dynamic programming only on $\ritem = M\setminus M^{\rmat}$, i.e. the unallocated items in step 1. The start state is $(\nvec_0^{\rmat}, \mat_0^{\rmat}, 0)$ which represent the state where all items in $M^{\rmat}$ and no items in $\ritem$ has been allocated. 

\paragraph{Step 4: searching} 
Finally, we search if there is a state that satisfies two conditions. (1) The state satisfies the quota $Q$. (2) For any agent $i$ and $h$, $i$ does not envy $h$ by more than $\frac{\ub_i}{\rd}$ under the rounded valuation $\ut_i^{\rmat}$, after removing one item, i.e. $\ut(A_i\setminus\{\gd\}) \ge \ut(A_i\setminus\{\gd\}) - \frac{\ub_i}{\rd}$ for some $\gd$.  
If such a state exists, the algorithm constructs the allocation backward and returns it. If such allocation does not exist for any state and any $\rmat$, the algorithm returns NO. The reason for searching a bounded envy allocation rather than an $\aefone$ allocation is to guarantee that the algorithm will always return NO if there does not exist an $\aefone$ allocation (Theorem~\ref{thm:dp2}). 


The following lemma shows that, by rounding the value function of each agent based on the most valuable item in $\ritem_i$, the average value of an agent's bundle is lower bound by the average value of $\ritem_i$. 

\begin{lem}
    \label{lem:dp2}
    Given any \rim{} $\rmat$ and allocation $A$, and for any agent $i$, if $i$ does not envy any other agent by more than $\varepsilon > 0$ under $\rmat$ and $\ut^{\rmat}$, then $u^{\rmat}_i(A_i) \ge \frac1{n} u^{\rmat}_i(\ritem_i) -\varepsilon$, where $\ritem_i$ is the set of $i$'s removing items. 
\end{lem}

\begin{proof}[Proof Sketch]
    Suppose this is not true, and agent $i$ has a bundle $A_i$ with average value less than $\frac1{n} u^{\rmat}_i(\ritem_i) -\varepsilon$. We consider the agent $h$ that takes the share of $\ritem_i$ with the largest average value under $i$'s valuation. Then $i$ envies $h$ by more than $\varepsilon$ even after removing one item, which is a contradiction. The full proof is in Appendix~\ref{apx:dp2_thm}.
\end{proof}

Lemma~\ref{lem:dp2} guarantees that the rounding error is small compared with the bundle's average value. 
Note that $\ritem_i$ contains at most $m$ items, and the largest item has a value of $\ub_i$. Therefore, (in $i$'s valuation,) the average value of $\ritem_i$ is at least $\frac{\ub_i}{m}$, and the average value of $A_i$ is at least $\frac{\ub_i}{mn} - \varepsilon$. On the other hand, the rounding error is $\frac{\ub_i}{\rd} = \frac{\ub_i}{m^2n^2}$. With a lower bound of average value and a bounded error, a reasonably good approximation ratio can be guaranteed, as shown in Theorem~\ref{thm:dp2}. 

\begin{thm}
\label{thm:dp2}
Given any instance of $\aefoneexst$ with quota $(\inst, Q)$, 
\begin{enumerate}
    \item if the algorithm returns NO, then $(\inst, Q)$ does not have an $\aefone$ allocation satisfying $Q$. 
    \item if the algorithm returns YES, it gives a $\aefonemul{(1 - \frac{4}{mn})}$ allocation satisfying $Q$.
\end{enumerate}  
\end{thm}

\begin{proof}[Proof Sketch]
(NO case) We turn to prove the equivalent statement that if $(\inst, Q)$ exists an $\aefone$ allocation satisfying $Q$, then the algorithm always returns YES. Suppose $A$ is an $\aefone$ allocation satisfying $Q$, and $\rmat$ is a \rim{} of $A$ that achieves $\aefone$. We show that any agent $i$ will not envy another agent $h$ by more than $\frac{\ub_i}{\rd}$ under $\ut^{\rmat}$, after removing one item. 
For the original valuation $\ut$, we have $\ut_i(A_i\setminus\{\gd\}) \ge \ut_i(A_h\setminus \{\gd\})$ for any $i, h$ and some $\gd$. For the rounded valuation, we  have $\vt_i(\gd)\le \vt_i^{\rmat}(\gd)\le \vt_i(\gd) + \frac{\ub_i}{\rd}$. Therefore, the average value of any bundle will neither decrease nor increase more than $\frac{\ub}{\rd}$ in the rounded valuation. Therefore, 
$\ut^{\rmat}_i(A_i\setminus\{\gd\}) -\ut_i^{\rmat}(A_h\setminus \{\gd\})\ge -\frac{\ub_i}{\rd}$, and agent $i$ does not envy agent $h$ by more than $\frac{\ub_i}{\rd}$. In the process of the algorithm, the state of $A$ will be validated in Step 3 and found in Step 4, and the algorithm will return YES. 

(YES case) We show that if an allocation $A$ satisfies that any agent $i$ will not envy another agent $h$ by more than $\frac{\ub_i}{\rd}$ under $\ut^{\rmat}$ after removing one item, then $A$ will be $\aefonemul{(1 - \frac{4}{mn})}$ under $\ut$. 
Consider any pair of agent $i$ and $h$ and suppose $i$ still envies $h$ even after removing one item under the original valuation $\ut$. (If such a pair does not exist, then $A$ is an $\aefone$ allocation, and the statement holds.)
For simplicity, let $A_i'$ and $A_h'$ be the bundles after agent $i$ removing item $\gd$. Since envy is bounded by $\frac{\ub_i}{\rd}$ in the rounded valuation $\ut^{\rmat}$, it is also bounded in the original valuation: $\ut_i(A_i') \ge \ut_i(A_h') - \frac{2\ub_i}{\rd}$. On the other hand, since $i$ envies $h$, $\ut_i(A_h') > \ut_i(A_i')$. Consider two subcases. 
\begin{enumerate}
    \item $\ut_i(A_i') \ge \ut_i(A_i)$. In this case agent $i$ either remove an item in $A_h$ or an item with a value lower than average in $A_i$. In this case, we have $\ut_i(A_h') > \ut_i(A_i)$. With the rounding, we know $\ut_i(A_i) \ge \ut_i^{\rmat}(A_i) - \frac{\ub_i}{\rd}$. And from Lemma 2, $\ut_i^{\rmat}(A_i) \ge \frac1{n} u^{\rmat}_i(\ritem_i) -\frac{\ub_i}{\rd} \ge \frac{\ub_i}{mn} - \frac{\ub_i}{\rd}$.  Aggregating all these inequalities, we get     
    \begin{align*}
        \ut_i(A_i') \ge & \ut_i(A_h') - \frac{2\ub_i}{\rd} \\
        = &(1 - \frac{1}{4mn})\ut_i(A_h') +  \frac{1}{4mn}\ut_i(A_h')- \frac{2\ub_i}{\rd}\\
        \ge &(1 - \frac{1}{4mn})\ut_i(A_h') +  \frac{1}{4mn}\cdot \frac{\ub_i}{mn} -\frac{4\ub_i}{\rd}\\
        = & (1 - \frac{1}{4mn})\ut_i(A_h'). 
    \end{align*}
    \item $\ut_i(A_i') < \ut_i(A_i)$. In this case, $i$ removes an item in $A_i$ with a value higher than average. With a similar reasoning, we turn to show that  $\ut_i(A_i) \ge (1 - \frac{1}{4mn})\ut_i(A_h)$.
\end{enumerate}

Therefore, $A$ is an $\aefonemul{(1 - \frac{4}{mn})}$ allocation under the original valuation $\ut$. 
\end{proof}

}

\section{Conclusion and Future Work}
In this paper, we propose average envy-freeness where envy is defined by the average value of a bundle. AEF provides a fairness criterion for allocation problems in collaboration scenarios, where agents have different entitlements, and the entitlements depend on the allocation itself. We study the existence and complexity of AEF and AEF-1. While deciding the existence of an AEF allocation is NP-hard, an AEF-1 allocation always exists and can be computed in polynomial time. We also study the complexity of AEF-1 with quota. While AEF-1 with quota is NP-complete to decide, we provide polynomial-time algorithms for instances with a constant number of agents to find AEF-1 allocation under binary valuation and approximated AEF-1 allocation under general valuation. 


The notion of average envy-freeness can be extended in multiple aspects.
One extension is scenarios with multiple copies. For example, in a paper review scenario, a paper should be reviewed by multiple papers, but a reviewer cannot review a paper multiple times. Another extension is scenarios where items bring different entitlements to agents. We expect the entitlement of agents to be the sum of entitlements of items in their bundles, and envy is defined on the sum of values divided by the entitlement of the agent. It is also an intriguing direction to find relaxations of AEF other than AEF-1. An interesting observation is that ``AEF-X'' is an even stronger notion than AEF since agents should not envy each other even if their bundle's average value is decreased by removing the most valuation item.

\bibliographystyle{named}
\bibliography{references,ref_new}

\clearpage
\onecolumn
\appendix
{\section{Full Proofs}
\label{apx:proof}
\subsection{Theorem~\ref{thm:aef}}
\label{apx:aef}
\begin{thmbis}{thm:aef}
$\aefexst$ is NP-complete even for two agents and under identical valuation. 
\end{thmbis}
\begin{proof}
    We construct a reduction from $\partition$, which is known to be NP-complete~\citep{Garey79:Computers}. An instance of $\partition$ consists of a multiset $X = \{x_1, x_2,\cdots, x_{k}\}$ where $x_i\in \mathbb{N}$. The goal is to determine whether there exists a subset $Y\subset X$ such that $\sum_{x_i\in Y}x_i = \sum_{x_i\in X\setminus Y}x_i = T$, where $T = \frac12\sum_{x_i\in X}x_i$. Without the loss of generality, we assume that $k \ge 4$ and $T\ge 4$.

    Given an instance of $\partition$, we construct an $\aefexst$ instance as follows. There are $n=2$ agents and $m = 2k$ items. $M = \{\gds_i, \gdl_i\}_{i=1,2,\cdots, k}$. Two agents share the same valuation function $\vt$. $\ut$ is the average value function on $\vt$. For each $i$, $\vt(\gds_i) = (T^2k^2)^i$, and $\vt(\gdl_i) = (T^2k^2)^i + x_i$. 

    Suppose $\partition$ is a YES instance, and $Y\subset X$ and $X\setminus Y$ is an equal-sum partition of $X$. Consider the following allocation $A$: $A_1 = \{\gdl_i \mid x_i\in Y\}\cup \{\gds_i \mid x_i\in X\setminus Y\}$, and $A_2 = M\setminus A_1$. It's not hard to verify that $\ut(A_1) = \ut(A_2) = \frac{T + \sum_{i=1}^k (T^2k^2)^i}{k}$, which indicates AEF. 

    Suppose $\aefexst$ is a YES instance, and $A$ is an $\aef$ allocation. We show that $A$ induces a partition of $X$ in the following four steps. 

    First, $\gds_{k}$ and $\gdl_{k}$ must be allocated to different agents in $A$. Suppose this is not the case and agent 1 gets both $\gds_{k}$ and $\gdl_{k}$. We show that agent 2 will envy agent 1. 
    Note that $\vt(A_1) \ge 2(T^2k^2)^{k} + x_k$. Therefore, $\ut(A_1) \ge \frac{\vt(A_1)}{2k} \ge T^{2k}k^{2k-1}$. On the other hand, $\ut(A_2)$ cannot exceed the value of the most valuable item left, i.e. $\ut(A_2)\le \vt(\gdl_{k-1}) \le (T^2k^2)^{k-1}+T$. Therefore, $\ut(A_1) > \ut(A_2)$, and which is a contradiction.  
    


    Second, $A_1$ and $A_2$ both contain exactly $k$ items. Suppose this is not the case, and $|A_1| < |A_2|$. Then we have $|A_1| \le k-1$ and $|A_2| \ge k+1$. 
    We show that agent 2 must envy agent 1. 
    Since $A_1$ contains either $\gds_{k}$ or $\gdl_{k}$, $\vt(A_1) \ge (T^2k^2)^{k}$, and $\ut(A_1)\ge \frac{(T^2k^2)^{k}}{k-1}$. $A_2$, on the other hand, can have at most all items expect for $\gds_{k}$. Therefore, $\vt(A_2) \le 2T + \sum_{i=1}^k (T^2k^2)^i - (T^2k^2)^k$, and $\ut(A_2) \le \frac{\vt(A_2)}{k+1}$. It's not hard to verify that $\ut(A_1) > \ut(A_2)$, which is a contradiction. As $A$ is an $\aef$ allocation, $|A_1| = |A_2|$ implies that $\vt(A_1) = \vt(A_2)$.
    

    Third, for any $i = 1,2\cdots, k$, $A_1$ (and $A_2$, respectively) contains exactly one of $\gds_i$ and $\gdl_i$. Suppose this is not the case, and $i$ is the largest index such that $\gds_i$ and $\gdl_i$ are allocated to the same agent (suppose agent 1). We show that agent 2 will envy agent 1.
    $A_1$ contains one of $\gds_j$ and $\gdl_j$ for each $j = i+1, i+2,\cdots k$, and both $\gds_i$ and $\gdl_i$. 
    Therefore, $\vt(A_1) \ge \sum_{j = i+1}^{k} (T^2k^2)^{j} + 2(T^2k^2)^{i}$.
    On the other hand $\vt(A_2)$ will not exceed the sum of the rest value. Therefore, $\vt(A_2) \le \sum_{j = i+1}^{k} (T^2k^2)^{j} + (k-i+1)(T^2k^2)^{i-1} + 2T$. We can verify that $\vt(A_1) > \vt(A_2)$, Since both bundles contains $k$ items, $\ut(A_1) - \ut(A_2) = \frac{1}{k}(\vt(A_1) - \vt(A_2)) > 0$, which is a contradiction. 

     Finally, we show that $A_1$ and $A_2$ induce a partition on $X$. Let $Y = \{x_i \mid \gdl_i \in A_1\}$, we show that $Y$ and $X\setminus Y$ is a partition of $X$. Note that $\vt(A_1) = \sum_{i=1}^{k} (T^2k^2)^{i} +\sum_{x_i\in Y} x_i$, and $\vt(A_2) = \sum_{i=1}^{k} (T^2k^2)^{i}+ \sum_{x_i\in (X\setminus Y)} x_i$. To achieve AEF, we have $\vt(A_1)=\vt(A_2)$, which implies $\sum_{x_i\in Y} x_i = \sum_{x_i\in (X\setminus Y)} x_i$. Therefore, $Y$ and $X\setminus Y$ is a partition of $X$, and {\sc Partition} is a YES instance. 
\end{proof}

\subsection{Theorem~\ref{thm:aefoneq}}
\label{apx:aefoneq}
\begin{thmbis}{thm:aefoneq}
    $\aefoneqexst$ with fixed $n\ge 3$ is NP-complete. 
\end{thmbis}
\begin{proof}
    We show a reduction from {\sc Equal-cardinality Partition}, a variation of $\partition$ that requires equal size between two subsets and is also NP-complete~\citep{Garey79:Computers}. 
    An instance of {\sc Equal-cardinality Partition} consists of a multiset $X = \{x_1, x_2,\cdots, x_{2k}\}$ where $x_i\in \mathbb{N}$. The goal is to determine whether there exists a subset $Y\subset X$ of size $k$ such that $\sum_{x_i\in Y}x_i = \sum_{x_i\in X\setminus Y}x_i = T$, where $T = \frac12\sum_{x_i\in X}x_i$. With the loss of generality, we assume that $k \ge 4$ and $T\ge 4$. 

     Given an {\sc Equal-cardinality Partition} instance, we construct a $\aefoneqexst$ instance with three agents and $3k+6$ items. If $n > 3$, we add agents that value all items as 0 and are required to receive no items by the quota.
     
     Agents share the same valuation function $\vt$, and the quota $Q$ requires every agent to have exactly $k+2$ items. 
     The item set $M = M_1\cup M_2\cup M_3$ consists of three parts:
    \begin{itemize}
        \item $M_1 = \{\gd_1, \gd_2, \cdots, \gd_{2k}\}$, where $\vt(\gd_j) = x_j + k^2T^2$. \\Let $T' =\frac12 \sum_{\gd\in M_1} \vt(\gd) =T + k^3T^2.$
        \item $M_2$ contains $k+1$ copies of $b$ with $\vt(b) = \frac{(k+2)T'}{(k+1)^2}.$
        \item $M_3$ contains five copies of $0$ with $\vt(0) = 0$. 
    \end{itemize}

     We state that the value of $b$ is smaller than any item in $M_1$. The proof of the lemma is at the end of this proof. 
    \begin{lembis}{lem:aefq}
        For any $g\in M_1$, $\vt(b) < \vt(\gd).$
    \end{lembis}

    If {\sc Equal-cardinality partition} is a YES instance, and $Y$ and $X\setminus Y$ are a solution, we show the following allocation $A$ is $\aefone$ and satisfies $Q$.
    \begin{enumerate}
        \item $A_1 = \{g_j\mid x_j\in Y\} \cup \{0,0\}$.
        \item $A_2 = \{g_j\mid x_j\in (X\setminus Y)\}\cup \{0,0\}$.
        \item $A_3 = M_2\cup \{0\}$. 
    \end{enumerate}
    It's not hard to verify that each agent gets exactly $k+2$ items, $\ut(A_1) = \ut(A_2) = \frac{T'}{k+2}$, and $\ut(A_3) = \frac{T'}{k+1}$. Agent 1 and 2 does not envy each other, and agent 3 does not envy agent 1 and 2. When comparing with agent 3, agent 1 and agent 2 can remove a 0 in their own bundle, and $\ut(A_1\setminus \{0\}) = \ut(A_2\setminus \{0\}) = \ut(A_3) = \frac{T'}{k+1}$. Therefore, such allocation is $\aefone$ and satisfies $Q$. 

    If $\aefoneqexst$ is a YES instance, and $A$ is an $\aefone $ allocation satisfying $Q$. We show that {\sc Equal-cardinality partition} is a YES instance in the following steps. 

    First, no agent can have more than two item $0$ in their bundles. Suppose this is not the case, and agent 1 gets at least three 0. Without loss of generality, assume agent 2 gets at most one 0. We show that agent 1 envies agent 2 even after removing one item. 
    \begin{itemize}
        \item If agent 1 removes an item in $A_1$, the best choice is to remove an $0$. Then the average value of $A_1\setminus \{0\}$ is at most $\frac{2T + (k-1)k^2T^2}{k+1}$. This is achieved when $A_1$ contains exactly three 0 and all $\gd_j$ such that $x_j > 0$ (Recall that $\vt(\gd) > \vt(b)$, so containing $b$ will decrease the average value). On the other hand, the average value of $A_2$ is at least $\frac{(k+1)\cdot \vt(b)}{k+2}$ for one item $0$ and $k+1$ item $b$. Then we have $\ut(A_1\setminus\{0\}) < \ut(A_2)$.
        \item If agent 1 removes an item $\gd$ in $A_2$, the average value of $A_1$ is at most $\frac{2T + (k-1)k^2T^2}{k+2}$, and the value of remaining $A_2$ is at least $\frac{k\cdot \vt(b)}{k+1}$. Still, $\ut(A_1) < \ut(A_2\setminus \{\gd\})$.
        \end{itemize}
    In both cases, agent 1 envies agent 2, which is a contradiction. Therefore, no agents can have more than two item $0$. Then there are two agents with two item $0$ and one agent with one item $0$.  

    Second, the agent with exactly one item $0$ must have all the item $b$. Without loss of generality, let agent $3$ have exactly one 0. Then $A_3 = M_2\cup\{0\}$. 
    Suppose this is not the case, and $A_3$ contains at least one item from $M_1$. Since $\vt(\gd) > \vt(b)$, we have $\ut(A_3) > \frac{T'}{k+1}$, and at least one of $A_1$ and $A_2$ will have a value less than $\frac{T'}{k+2}$. Without loss of generality, assume $\ut(A_1) < \frac{T'}{k+2}$. We show that agent 1 envies agent 3 even after removing one item. 
    \begin{itemize}
        \item If agent 1 removes an item in $A_1$, the best choice is to remove an item $0$. Then $\ut(A_1\setminus\{0\}) = \frac{k+2}{k+1}\ut(A_1) < \frac{T'}{k+1} < \ut(A_3).$
        \item If agent 1 removes an item $\gd$ in $A_3$, then $\ut(A_3\setminus\{\gd\}) \ge \frac{k\cdot \vt(b)}{k+1} = \frac{k(k+2)T'}{(k+1)^3} = \frac{T'}{k+1} - \frac{T'}{(k+1)^3} > \frac{T'}{k+2}.$ Still $\ut(A_1) < \ut(A_3\setminus \{\gd\})$. 
    \end{itemize}
    Agent 1 envies agent 3 even after removing one item, which is a contradiction. Therefore, $A_3$ must contain all the item $b$. 

    Finally, agent 1 and agent 2's bundles must derive a partition of $X$. That is $A_1 = M_1'\cup\{0,0\}$, $A_2 = (M_1\setminus M_1')\cup \{0,0\}$, where $Y = \{x_j \mid \gd_j\in M_1'\}$ and $X\setminus Y$ is a partition of $X$. Suppose it is not the case, and $Y$ is not a partition of $X$. Without loss of generality, assume $\sum_{x_j\in Y} x_j < T$. Then we have $\ut(A_1) < \frac{T'}{k+2}$. With a similar reasoning to the second step, we can show that agent 1 envies agent 3 even after removing one item. Therefore, $Y$ must be a partitioning of $X$, and {\sc Equal-cardinality partition} is a YES instance.
\end{proof}

\begin{proof}[Proof of Lemma~\ref{lem:aefq}]
    Note that $\vt(\gd) \ge k^2T^2$ and $\vt(b) = \frac{(k+2)T'}{(k+1)^2} = \frac{(k+2)(T + k^3T^2)}{(k+1)^2} = $. Therefore,
    \begin{align*}
        \vt(\gd) - \vt(b) \ge& k^2T^2 - \frac{(k+2)(T + k^3T^2)}{(k+1)^2}\\
        = & (1 - \frac{k(k+2)}{(k+2)^2})k^2T^2 - \frac{(k+2)T}{(k+1)^2}\\
        =& \frac{k^2T^2}{(k+1)^2}- \frac{(k+2)T}{(k+1)^2}\\
        = & \frac{k^2T^2 - (k+2)T}{(k+1)^2}\\
        > & 0.
    \end{align*}
\end{proof}

\subsection{Theorem~\ref{thm:dp2}}
\label{apx:dp2_thm}
\begin{thmbis}{thm:dp2}
Given any instance of $\aefoneexst$ with quota $(\inst, Q)$, 
\begin{enumerate}
    \item if Algorithm~\ref{alg:dp2} returns NO, then $(\inst, Q)$ does not have an $\aefone$ allocation satisfying $Q$. 
    \item if Algorithm~\ref{alg:dp2} returns YES, it gives a $\aefonemul{(1 - \frac{4}{mn})}$ allocation satisfying $Q$.
\end{enumerate}  
\end{thmbis}

\begin{proof}
For the proof of the theorem, we first propose the following lemma.

\begin{lembis}{lem:dp2}
Given any \rim{} $\rmat$ and allocation $A$, and for any agent $i$, if $i$ does not envy any other agent by $\varepsilon > 0$ under $\rmat$ and $\ut^{\rmat}$, then $u^{\rmat}_i(A_i) \ge \frac1{n} u^{\rmat}_i(\ritem_i) -\varepsilon$.  
\end{lembis}

Lemma~\ref{lem:dp2} gives a lower bound for agents' valuation on their own bundle, which guarantees the approximation ratio. The proof of Lemma~\ref{lem:dp2} is at the end of the proof. 

For the NO case, we turn to prove that if $(\inst, Q)$ exists an $\aefone$ allocation satisfying $Q$, then Algorithm~\ref{alg:dp2} always returns YES. Suppose $A$ is an $\aefone$ allocation satisfying $Q$, and $\rmat$ is the corresponding \rim. Then for any agent $i\neq h$ and the corresponding removing item $\gd$, $\ut_i(A_i\setminus\{\gd\}) \ge \ut_i(A_h\setminus \{\gd\})$. 
After rounding to $\ut^{\rmat}$, the average value of $(A_i\setminus\{\gd\})$ will not decrease, and the average value of $(A_h\setminus\{\gd\})$ will increase no more than $\frac{\ub_i}{\rd}$ according to the rounding on $\vt_i$. Therefore,
\begin{equation*}
    \ut^{\rmat}_i(A_i\setminus\{\gd\}) -\ut_i^{\rmat}(A_h\setminus \{\gd\})\ge -\frac{\ub_i}{\rd}.
\end{equation*}
In $A$, agent $i$ will not envy another agent by more than $\frac{\ub_i}{\rd}$ after removing one item under$\ut^{\rmat}$. Therefore, the state corresponding to $A$ will be discovered when searching under $\rmat$. 

For the YES case, we show that if $A$ satisfies that any agent $i$ will not envy another agent by more than $\frac{\ub_i}{\rd}$ after removing one item under $\ut^{\rmat}$, then $A$ will be $\aefonemul{(1 - \frac{4}{mn})}$ under $\ut$. From Lemma~\ref{lem:dp2} we know that $\ut^{\rmat}_i(A_i) \ge \frac1{n} \ut^{\rmat}_i(\ritem_i) -\frac{\ub_i}{\rd}$. Suppose $i$ still envy $h$ even  after removing on item $\gd$ under $\vt$. (If this case does not exist, $A$ is an $\aefone$ allocation, and the statement holds).
We discuss cases where $\gd$ belongs to different bundles. 

\noindent{\bf Case 1:} $\gd \in A_h$. 
By the envy guarantee in the rounded valuation, we have 
\begin{align*}
    \ut_i(A_i) \ge \ut_i^{\rmat}(A_i) - \frac{\ub_i}{\rd}
    \ge  \ut_i^{\rmat}(A_h\setminus \{\gd\}) - \frac{2\ub_i}{\rd}
    \ge  \ut_i(A_h\setminus \{\gd\}) - \frac{2\ub_i}{\rd}
\end{align*}
Then by the envy in the original valuation, we have 
\begin{align*}
    \ut_i(A_h\setminus \{\gd\}) > \ut_i(A_i) \ge \ut_i^{\rmat}(A_i)- \frac{\ub_i}{\rd} \ge \frac1{n}\ut_i^{\rmat}(\ritem_i)- \frac{2\ub_i}{\rd}
\end{align*}

Now note that $\ritem_i$ contains at most $m$ items, of which the largest valuation is $\ub_i$. Therefore, $\ut_i^{\rmat}(\ritem_i) \ge \frac{\ub_i}{m}$. Therefore, 
\begin{align*}
    \ut_i(A_i) \ge & \ut_i(A_h\setminus \{\gd\}) - \frac{2\ub_i}{\rd}\\
    = & (1 - \frac{4}{mn})\ut_i(A_h\setminus \{\gd\}) + \frac{4}{mn}\cdot \ut_i(A_h\setminus \{\gd\}) - \frac{2\ub_i}{\rd}\\
    \ge & (1 - \frac{4}{mn})\ut_i(A_h\setminus \{\gd\}) + \frac{4}{mn}\cdot (\frac1{n}\ut_i^{\rmat}(\ritem_i)- \frac{2\ub_i}{\rd}) - \frac{2\ub_i}{\rd}\\
    \ge & (1 - \frac{4}{mn})\ut_i(A_h\setminus \{\gd\}) + \frac{4}{mn}\cdot \frac{a}{mn} - \frac{4\ub_i}{m^2n^2} \\
    \ge & (1 - \frac{4}{mn})\ut_i(A_h\setminus \{\gd\}). 
\end{align*}

\noindent{\bf Case 2:}  $\gd \in A_i$, and $\ut_i^{\rmat}(A_i\setminus\{\gd\}) \ge \ut_i^{\rmat}(A_i)$. With a similar reasoning, we have $\ut_i(A_i\setminus\{\gd\}) \ge \ut_i(A_h) - \frac{2\ub_i}{\rd}$.  And 
\begin{align*}
    \ut_i(A_h) > \ut_i(A_i\setminus\{\gd\}) \ge \ut_i(A_i) \ge \ut_i^{\rmat}(A_i)- \frac{\ub_i}{\rd} \ge \frac1{n}\ut_i^{\rmat}(\ritem_i)- \frac{2\ub_i}{\rd}
\end{align*}
Therefore, 
\begin{align*}
    \ut_i(A_i\setminus\{\gd\}) \ge &\ut_i(A_h) - \frac{2\ub_i}{\rd} \\
    \ge & (1 - \frac{4}{mn})\ut_i(A_h) + \frac{4}{mn}\cdot (\frac1{n}\ut_i^{\rmat}(\ritem_i)- \frac{2\ub_i}{\rd}) - \frac{2\ub_i}{\rd}\\
    \ge & (1 - \frac{4}{mn})\ut_i(A_h).
\end{align*}

\noindent{\bf Case 3:}  $\gd \in A_i$, and $\ut_i^{\rmat}(A_i\setminus\{\gd\}) < \ut_i^{\rmat}(A_i )$. In this case, we have 
\begin{align*}
    \ut_i(A_i) > \ut_i(A_i\setminus \{\gd\}) \ge \ut_i^{\rmat}(A_i\setminus \{\gd\}) - \frac{\ub_i}{\rd}
    \ge  \ut_i^{\rmat}(A_h) - \frac{2\ub_i}{\rd}
    \ge  \ut_i(A_h) - \frac{2\ub_i}{\rd}.
\end{align*}
Note that $\ut^{\rmat}(A_i) \ge \frac1{n}\ut^{\rmat}(\ritem_i) - \frac{\ub_i}{\rd}$, which implies $\ut(A_i) \ge \frac1{n}\ut^{\rmat}(\ritem_i) - \frac{2\ub_i}{\rd}$. Therefore, 
\begin{align*}
    \ut_i(A_i) > & (1 - \frac{4}{mn})(\ut_i(A_h) -\frac{2\ub_i}{\rd}) + \frac{4}{mn}\ut_i(A_i)\\
    \ge & (1 - \frac{4}{mn})(\ut_i(A_h) -\frac{2\ub_i}{\rd}) + \frac{4}{mn}(\frac1{n}\ut^{\rmat}(\ritem_i) - \frac{2\ub_i}{\rd})\\
    \ge & (1 - \frac{4}{mn})\ut_i(A_h) + \frac{4}{mn} \cdot \frac{\ub_i}{mn} - \frac{4\ub_i}{r}\\
    \ge & (1 - \frac{4}{mn})\ut_i(A_h).
\end{align*}

Therefore, we show that $A$ is an $\aefonemul{(1-\frac{4}{mn})}$ allocation under $\ut$. 

\end{proof}

\begin{proof}[Proof of Lemma~\ref{lem:dp2}]
    Suppose this is not true, and there exists an agent $i$ such that $\ut^{\rmat}_i(A_i) < \frac1{n} \ut^{\rmat}_i(\ritem_i) -\varepsilon$. In this case $\frac1{n}\ut^{\rmat}_i(\ritem_i) > \varepsilon$. 
    
    First we show that $\ut^{\rmat}_i(A_i\cap \ritem_i) < \ut^{\rmat}_i(\ritem_i)$. This is because $A_i$ contains at most $n-1$ items from $M^{\rmat}_i$ (one for every other agent). If $A_i\cap \ritem_i = 0$, then $\ut^{\rmat}_i(A_i\cap \ritem_i) = 0$. Otherwise, $\ut^{\rmat}_i(A_i\cap \ritem_i) \le n \ut^{\rmat}_i(A_i) < \ut^{\rmat}_i(\ritem_i)$. Then there must exist another agent $h\neq i$ that takes the share of $\ritem_i$ with the largest average value in $i$'s valuation, i.e. $\ut^{\rmat}_i(A_h \cap \ritem_i) \ge \ut^{\rmat}_i(\ritem_i)$. We show $i$ must envy $h$. Suppose the item $i$ removes when comparing with $h$ is $\gd$ (which is determined by $\rmat$). 

    \begin{enumerate}
        \item If $\gd\in A_h$, then $A_h\setminus\{\gd\}$ does not contain any item from  $M^{\rmat}_i$ and at least one item from $\ritem_i$. Therefore, \begin{equation*}
            \ut^{\rmat}_i(A_h\setminus\{\gd\}) = \ut^{\rmat}_i(A_h \cap \ritem_i) \ge \ut^{\rmat}_i(\ritem_i) > \ut^{\rmat}_i(A_i) + \varepsilon.
        \end{equation*}
        \item If $\gd\in A_i$, then $A_h$ contains at most $1$ items from $M^{\rmat}_i$ and at least one item from $\ritem_i$. Therefore, 
        \begin{equation*}
            \ut^{\rmat}_i(A_h) \ge \frac{1}{2} \ut^{\rmat}_i(A_h \cap \ritem_i) \ge \frac{1}{2}\ut^{\rmat}_i(\ritem_i)
        \end{equation*}
        
        For $i$, if $A_i$ contains exactly one item, Then 
        \begin{equation*}
            \ut^{\rmat}_i(A_i\setminus\{\gd\}) = 0 < \frac1{n}\ut^{\rmat}_i(\ritem_i) - \varepsilon \le \ut^{\rmat}_i(A_h) - \varepsilon.
        \end{equation*}
        
        Otherwise, $A_i$ contains at least two items. Therefore, 
        \begin{equation*}
            \ut^{\rmat}_i(A_i\setminus\{\gd\}) \le 2\ut^{\rmat}_i(A_i) < 2\left(\frac1{n} \ut^{\rmat}_i(\ritem_i) -\varepsilon\right) < \ut^{\rmat}_i(A_h) - \varepsilon.
        \end{equation*}
        \item If $\gd = \emptyset$ ($i$ does not remove any item), similar to the $\gd\in A_i$ case, we have 
        \begin{equation*}
            \ut^{\rmat}_i(A_i) < \frac1{n} \ut^{\rmat}_i(\ritem_i) -\varepsilon \le \ut^{\rmat}_i(A_h) - \varepsilon.
        \end{equation*}
    \end{enumerate}
    Therefore, $i$ envies $h$ more than $\varepsilon$ even after removing one item, which is a contradiction.
\end{proof}

\section{Approximation Algorithm}
\label{apx:dp2_alg}
\begin{algorithm*}[htbp]

\caption{Approximated DP for $\aefone$ with quota}
\label{alg:dp2}
\begin{algorithmic}[1]
\REQUIRE Agent set $N$, item set $M$, valuation profile $V$, and quota $Q$. 
\ENSURE $\aefonemul{(1-\frac{4}{mn})}$ allocation satisfying $Q$.
\FOR{all valid \rim{} $\rmat$}
\STATE For each agent $i$, let $M^{\rmat}_i$ be the items to remove by $i$ in comparisons (indicated by $\rmat$) and $\ritem_i \leftarrow M\setminus M^{\rmat}_i$.
\STATE Construct the new valuation $V^{\rmat}$ by rounding the valuations for items in $\ritem_i$ for each agent $i$.\\
$\rd \leftarrow m^2n^2$ and $\ub_i\leftarrow \max_{\gd\in \ritem_i} \vt_i(\gd)$ for all $i$.\\
For all agent $i$ and item $\gd\in M^{\rmat}_i$, $\vt^{\rmat}_i(\gd) \leftarrow \vt_i(\gd)$.\\
For all agent $i$ and item $\gd \in \ritem_i,$ $\vt^{\rmat}_i(\gd)$ is the rounded version of $\vt_i(\gd)$ with parameters $(\rd, \ub_i)$.
\STATE Allocate all items in $M^{\rmat}_i$ as $\rmat$ indicates.Let $\ritem'$ be the set of unallocated items. Let $\nvec^{\rmat}_0$ and $\mat^{\rmat}_0$ be the (initial) state after allocation.
\STATE Define the search space.
$\nvecset^{\rmat} = \{0,1,\cdots, m\}^{n}$.\\
$\matset^R_i = \{0, \frac{a_i}r, \frac{2a_i}r,\cdots, |\ritem_i|a_i\}^{n}$, and
$\matset^{\rmat} = \{H^{\rmat}_0 + H^{\rmat} | H^{\rmat} \in \matset^R_1 \times \matset^R_2\times\cdots \times \matset^R_n \}$
\STATE Initialization: $\isvld(\nvec^{\rmat}_0, \mat^{\rmat}_0, 0) \leftarrow 1$.
\FOR{$k=0,1,\cdots, |\ritem|-1$}
\FOR{$\nvec\in \nvecset^{\rmat} $ and $\mat\in\matset^{\rmat}$ such that $\isvld(\nvec, \mat, k) = 1$}
\FOR{$i = 1,2,\cdots n$}
\STATE Update $\nvec'$ and $\mat'$ after assigning the $k$-th item in $\ritem$ to agent $i$.
\STATE $\isvld(\nvec', \mat', k+1) \leftarrow 1$. $\prev(\nvec', \mat', k+1) \leftarrow i$.  
\ENDFOR
\ENDFOR
\ENDFOR
\FOR{$\nvec\in\nvecset^{\rmat}, \mat\in\matset^{\rmat}$ such that $\isvld(\nvec, \mat, |\ritem|) = 1$}
\IF{$(\nvec, \mat, |\ritem|)$ \\
(1) satisfies for each $i$, agent $i$ envies any other agent by at most $\frac{\ub_i}{\rd}$ after removing one item, \\
(2) satisfies the quota $Q$,}
\STATE Construct the allocation from $\prev$ backward. \textbf{return} YES and the allocation.
\ENDIF
\ENDFOR
\ENDFOR
\RETURN NO.
\end{algorithmic}
\end{algorithm*}}

\end{document}